\documentclass{acm_proc_article-sp}

\usepackage{multirow}
\usepackage{rotating}
\usepackage{enumerate}
\usepackage{amsfonts}
\usepackage{amsmath}
\usepackage{amssymb}
\usepackage{graphicx}
\usepackage{subfigure}
\usepackage{theorem}

\usepackage{graphicx} 
\usepackage{subfigure} 
\usepackage{amsmath,amssymb,graphicx,multirow,theorem,color}
\usepackage{color}
\usepackage{float}
\usepackage{comment}
\usepackage{flushend}

\usepackage{algorithm}
\usepackage{algorithmic}

\newfloat{algorithm}{t}{lop}

\newtheorem{theorem}{Theorem}[section]
\newtheorem{lemma}[theorem]{Lemma}

\newtheorem{definition}{Definition}[section]
\newtheorem{claim}{Claim}[section]

\usepackage{hyperref}

\usepackage{xr-hyper}
\usepackage{zref-base}
\usepackage{atveryend}
\begin{document}
\frenchspacing
\title{Anomaly Detection in Dynamic Networks of Varying Size}

\author{
Timothy La Fond$^1$, Jennifer Neville$^1$, Brian Gallagher$^2$\\
       {$^1$Purdue University, $^2$Lawrence Livermore National Laboratory}\\
       {{\{tlafond,neville\}@purdue.edu}, bgallagher@llnl.gov} \\
}

\maketitle

\begin{abstract}

Dynamic networks, also called network streams, are an important data representation that applies to many real-world domains.  Many sets of network data such as e-mail networks, social networks, or internet traffic networks are best represented by a dynamic network due to the temporal component of the data.  One important application in the domain of dynamic network analysis is anomaly detection.  Here the task is to identify points in time where the network exhibits behavior radically different from a typical time, either due to some event (like the failure of machines in a computer network) or a shift in the network properties.  This problem is made more difficult by the fluid nature of what is considered "normal" network behavior.  The volume of traffic on a network, for example, can change over the course of a month or even vary based on the time of the day without being considered unusual.  Anomaly detection tests using traditional network statistics have difficulty in these scenarios due to their \textbf{Density Dependence}: as the volume of edges changes the value of the statistics changes as well making it difficult to determine if the change in signal is due to the traffic volume or due to some fundamental shift in the behavior of the network.  To more accurately detect anomalies in dynamic networks, we introduce the concept of \textbf{Density-Consistent} network statistics.  These statistics are designed to produce  results that reflect the state of the network independent of the volume of edges.  On synthetically generated graphs anomaly detectors using these statistics show a a 20-400\% improvement in the recall when distinguishing graphs drawn from different distributions.  When applied to several real datasets Density-Consistent statistics recover multiple network events which standard statistics failed to find, and the times flagged as anomalies by Density-Consistent statistics have subgraphs with radically different structure from normal time steps.

\end{abstract}

\section{Introduction} 

Network analysis is a broad field but one of the more important applications is in the detection of anomalous or critical events.  These anomalies could be a machine failure on a computer network, an example of malicious activity, or the repercussions of some event on a social network's behavior \cite{tartakovsky, siris}.  In this paper, we will focus on the task of anomaly detection in a dynamic network where the structure of the network is changing over time.  For example, each time step could represent one day's worth of activity on an e-mail network.  The goal is then to identify any time steps where the pattern of those communications seems abnormal compared to those of other time steps.  

As comparing the communication pattern of two network examples directly is complex, one simple approach is to summarize each network using a network statistic then compare the statistics.  A number of anomaly detection methods rely on these statistics \cite{snijders, banks,  mcculloh2}.  Another method is to use a network model such as ERGM \cite{snijders4}.  However, both these methods often encounter difficulties when the properties of the network are not static.  

A typical real-world network experiences many changes in the course of its natural behavior, changes which are not examples of anomalous events.  The most common of these is variation in the volume of edges.  In the case of an e-mail network where the edges represent messages, the total number of messages could vary based on the time 
or there could be random variance in the number of messages sent each day.  The statistics used to measure the network properties are usually intended to capture some other effect of the network than simply the volume of edges: for example, the clustering coefficient is usually treated as a measure of the transitivity.
However, the common clustering coefficient measure is {\em statistically inconsistent} as the density of the network changes.  Even on an Erdos-Renyi network, which does not explicitly capture transitive relationships, the clustering coefficient will be greater as the density of the network increases.  

When statistics vary with the number of edges in the network, it is not valid to compare different network time steps using those statistics unless the number of edges is constant in each time step.  A similar effect occurs with network models that employ features or statistics which are size sensitive: \cite{shalizi} show that ERGM models learn different parameters given subsets of the same graph, so even if the network properties are identical observing a smaller portion of the network leads to learning a different set of parameters.  

The purpose of this work is to analytically characterize statistics by their sensitivity to network density, and offer principled alternatives that are {\em consistent estimators}, which empirically give more accurate results on networks with varying densities.


The major contributions of this paper are:

\begin{itemize}
\item We prove that several commonly used network statistics are \textbf{Density Dependent} and poorly reflect the network behavior if the network size is not constant.
\item We offer alternative statistics that are \textbf{Density Consistent} which measure changes to the distribution of edges regardless of the total number of observed edges.
\item We demonstrate through theory and synthetic trials that anomaly detection tests using Density Consistent statistics are better at identifying when the distribution of edges in a network has changed.
\item We apply anomaly detection tests using both types of statistics to real data to show that Density Consistent statistics recover more major events of the network stream while Density Dependent statistics flag many time steps due to a change in the total edge count rather than an identifiable anomaly.
\item We analyze the subgraphs that changed the most in the anomalous time steps and demonstrate that Density Consistent statistics are better at finding local features which changed radically during the anomaly.
\end{itemize}

\section{Statistic-Based Anomaly \\ Detection}

A statistic-based anomaly detection method is any method which makes its determination of anomalous behavior using network statistics calculated on the graph examples. The actual anomaly detection process can be characterized in the form of a hypothesis test.  The network statistics calculated on examples demonstrating {\em normal} network behavior form the null distribution, while the statistic on the network being tested for anomalies forms the test statistic.  If the test statistic is not likely to have been drawn from the null distribution, we can reject the null hypothesis (that the test network reflects normal behavior) and conclude that it is anomalous.   

Let $G_t =  \{ V, E_t \}$ be a multigraph that represents a dynamic network, where $V$ is the node set and $E_t$ is the set of edges at time $t$, with $e_{ij,t}$  the number of edges between nodes $i$ and $j$ at time $t$.  The edges represent the number of interactions that occurred between the nodes observed within a discrete window of time.  As the number of participating nodes is relatively static compared to the number of communications, we will assume that the node set is a constant; in time steps where a node has no communications we will treat it as being part of the network but having zero edges.


Let us define some network statistic $S_k(G_t)$ designed to measure network property $k$ which we will use as the test statistic (e.g., clustering coefficient).  Given some set of time steps $t_{Lx} \in \{ t_{L1},...t_{Lmax} \}$ such that the set of graphs in those times $\{G_{t_{Lx}}\}$ are all examples of {\em normal} network behavior, $S_k$ is calculated on each of these learning set examples to estimate an empirical {\em null} distribution.  If $t_{test}$ is the time step we are testing for abnormality, then the value $S_k(G_{t_{test}})$ is the test statistic.  Given a specified $p$-value (referred to as $\alpha$), we can find threshold(s) that reject $p$ percentage of the null distribution, then draw our conclusion about whether to reject the null hypothesis (conclude an anomaly is present) if the test statistic $S_k(G_{t_{test}})$ falls out of those bounds.  For this work we will use a two-tailed Z-test with a $p$-value of $\alpha = 0.05$ with the thresholds $\phi_{lower}$ and $\phi_{upper}$.  Anomalous test cases where the null hypothesis is rejected correspond to {\em true positives}; normal cases where the null hypothesis is rejected correspond to {\em false positives}.  Likewise anomalous cases where the null is not rejected correspond to {\em false negatives} and normal cases where the null is not rejected correspond to {\em true negatives}. 

Deltacon~\cite{koutra} is an example of the statistic-based approach, as are many others \cite{mcculloh, mcculloh2, priebe, snijders2}. 
As these models also often incorporate network statistics, we will focus on the statistics themselves in this paper.  
Moreover, not all methods rely solely on statistics calculated from single networks: there are also {\em delta} statistics which measure the difference between two network examples.  Netsimile is an example of such a network comparison statistic~\cite{berlingerio}.  For a dynamic network, these delta statistics are usually calculated between the networks in two consecutive time steps.  In practical use for hypothesis testing, these delta statistics function the same as their single network counterparts.  

If the network properties being tested are not static with respect to the time $t$ this natural evolution may cause $S_k(G_t)$ to change over time regardless of anomalies which makes the null distribution invalid.  It is useful in these instances to replace the statistic with a detrended version $S*_k(G_t) = S_k(G_t) - f(t)$ where the function $f(t)$ is some fit to the original statistic values.  The paper by \cite{tlafond} describes how to do dynamic anomaly detection using a linear detrending function, but other functions can be used for the detrending.  This detrending operation does not change the overall properties of the statistic so for the remainder of the paper assume $S_k(G_t)$ refers to a detrended version, if appropriate.

\subsection{Common Network Statistics}

Listed here are some of the more commonly used  network statistics for the anomaly detection. 

\noindent \textbf{Graph Edit Distance}:\\
The graph edit distance (GED) \cite{gao} is defined as the number of edits that one network needs to undergo to transform into another network and is a measure of how different the network topologies are.  In the context of a dynamic network the GED is applied as a delta measure to report how quickly the network is changing.

\vspace{-6mm}
\begin{small}
\begin{align} \label{eq:ged}
GED(G_{t}) &= | V_{t} | + | V_{t-1} | - 2 * | V_{t} \cap V_{t-1} | \nonumber \\ & + | E_{t} | +  | E_{t-1} | - 2 * | E_{t} \cap E_{t-1} |
\end{align}
\end{small}
\vspace{-6mm}

\noindent \textbf{Degree Distribution Difference}: \\
Define $D_{i,t} = \sum_{j \neq i} e_{ij,t}$ to be the degree of $i$, the total number of messages to or from the node.  The degree distribution is then a histogram of the number of nodes with a particular degree.  Typically real-world network exhibit a power-law degree distribution \cite{faloutsos} but others are possible.  To compare the degree distributions of $t$ and $t-1$, one option is to take the squared difference between the degree counts for each possible degree.  We will call this the degree distribution difference (DD):

\vspace{-5mm}
\begin{small}
\begin{align} \label{eq:dd}
DD(G_t) = \sum_{k=1}^{max_i(D_{i,t}, D_{i,t-1})} (\sum_i \mathbb I(D_{i,t} = k) - \mathbb I(D_{i,t-1} = k) )^2
\end{align}
\end{small}
\vspace{-5mm}

Other measures can be used to compare the two distributions but the statistical properties of the squared distance described later extend to other distance measures.

\noindent \textbf{Weighted Clustering Coefficient}:\\
Clustering coefficient is intended as a measure of the transitivity of the network.
It is a critical property of many social networks as the odds of interacting with friends of your friends often lead to these triangular relationships.  As the standard clustering coefficient is not designed for weighted graphs we will be analyzing a weighted clustering coefficient, specifically the Barrat weighted clustering coefficient (CB)\cite{saramaki}:

\vspace{-5mm}
\begin{small}
\begin{align} \label{eq:wcc}
CB(G_t) = \sum_{i} \frac{1}{N*(k_i-1)*s_i} \sum_{j,k} \frac{e_{i,j}+e_{i,k}}{2} a_{i,j}*a_{i,k}*a_{j,k} 
\end{align}
\end{small}
\vspace{-5mm}

\noindent Where $a_{i,j} = I[e_{i,j} > 0]$, $s_i = \sum_j e_{i,j}$, and $k_i = \sum_j a_{i,j}$.  Other weighted clustering coefficients exists but they behave similarly to the Barrat coefficient.

\section{Density Dependence}

To illustrate why dependency of the network statistic on the edge count affects the conclusion of hypothesis tests, we will us first investigate statistics that are {\em density dependent}.

\begin{definition}
A statistic $S$ is \textbf{density dependent} if the value of $S(G_t)$ is dependent on the density of $G_t$ (i.e., $| E_t |$).  
\end{definition}
\begin{theorem} \textbf{False Positives for Density Dependent Statistics} \\
Let $L \!=\! \{ G_{t_{Lx}} \}$ be a learning set of graphs and $G_{t_{test}}$ be the test graph.  If $S_k(G_t)$ is monotonically dependent on $| E_t |$ and $\{ | E_{t_{Lx}} | \}$ is bounded by finite $E_{lower}$ and $E_{upper}$, there is some $| E_{t_{test}} |$ that will cause the test case to be rejected regardless of whether the network is an example of an anomaly with respect to property $k$. 
\label{falsepos}
\end{theorem}

\begin{proof}
Let $S_k(G)$ be a network statistic that is a monotonic and divergent function with respect to the number of edges in $G$.  Given a set of learning graphs $\{G_{t_{L}}\}$ the values of $S_k(\{G_{t_{L}}\})$ are bounded by $max(S_k(\{G_{t_{L}}\}))$ and $min(S_k(\{G_{t_{L}}\}))$, so the critical points $\phi_{lower}$ and $\phi_{upper}$ of a hypothesis test using this learning set will be within these bounds.  Since an increasing $| E_{t_{test}} | $ implies $S_k(G_{t_{test}})$ increases or decreases, then there exists a $| E_{t_{test}} |$ such that $S_k(G_{t_{test}})$ is not within $\phi_{lower}$ and $\phi_{upper}$ and will be rejected by the test.
\end{proof}

If changing the number of edges in an observed network changes the output of the statistic, then if the test network differs sufficiently in its number of edges compared to the learning examples the null hypothesis will be rejected regardless of the other network properties.  As for why it is not sufficient to simply label these times as anomalies (due to unusual edge volume) the hypothesis test is designed to test for abnormality in a specific network property.  If edge count anomalies also flag anomalies on other network properties, we cannot disambiguate the case where both are anomalous or just the message volume.  If just the message volume is unusual, this might simply be an example of an exceptionally busy day where the pattern of communication is roughly the same just elevated.  This is a very different case from where both the volume and the distribution of edges are unusual.  

A second problem occurs when the edge counts in the learning set have high variance.  If the statistic is dependent on the number of edges, noise in the edge counts translates to noise in the statistic values which lowers the statistical power of the test.  

\begin{theorem} \textbf{False Negatives for Density Dependent Statistics} \\
For any $S_k(G_{t_{test}})$ calculated on a network that is anomalous with respect to property $k$, if $S_k(G_t)$ is dependent on $| E_t |$ there is some value of the variance of the learning network edge counts such that $S_k(G_{t_{test}})$ is not detected as an anomaly.
\label{falseneg}
\end{theorem}

\begin{proof}
Let $S_k(G_{t_{test}})$ be the test statistic and $\{G_{t_{L}}\}$ be the set of learning graphs where the edge count of any learning graph $ | E_{t_{Lx}} |$ be drawn according to distribution $\Phi_E$.  If $S_k(G_t)$ is a monotonic divergent function of $| E_t |$ then as the variance of $\Phi_E$ increases the variance of $S_k(G_{t_{Lx}})$ increases as well.  For a given $\alpha$, the hypothesis test thresholds $\phi_{lower}$ and $\phi_{upper}$ will widen as the variance increases to incorporate $1 - \alpha$ learning set instances.  Therefore, for a given $S_k(G_{t_{test}})$ there is some value of $var(\Phi_E)$ such that  $\phi_{lower} < S_k(G_{t_{test}}) < \phi_{upper}$.
\end{proof}

\noindent With a sufficient amount of edge count noise, the statistical power of the anomaly detector drops to zero.  

These theorems have been defined using a statistic calculated on a single network, but some statistics are delta measures which are measured on two networks.  In these cases, the edge counts of either or both of the networks can cause edge dependency issues.

\begin{lemma}
If a delta statistic $S_k(G_t, G_{t-1})$ is dependent on either $E_{min}(G_t) \!=\! min(| E_t |, | E_{t-1} |)$ or $E_{\Delta}(G_t) \!=\! abs(| E_t | \!-\! | E_{t-1} |)$ then Thms \ref{falsepos}-\ref{falseneg} apply to the delta statistic.
\label{deltalemma}
\end{lemma}
\vspace{-4mm}

\begin{proof}
For some delta statistic $S_k(G_t, G_{t-1})$ if the statistic is dependent on $| E_t |$, $| E_{t-1} |$, or both, Theorems \ref{falsepos} and \ref{falseneg} apply to any edge count which influences the statistic.  If $S_k(G_t, G_{t-1})$ depends on $E_\Delta = abs(| E_t | - | E_{t-1} |)$ then as either $| E_t |$ or $| E_{t-1} |$ change the statistic produced changes leading to the problems described in theorems \ref{falsepos} and \ref{falseneg}.  If $S_k(G_t, G_{t-1})$ depends on $E_{min} = min(| E_t |, | E_{t-1} |)$ then if both $| E_t |$ and $| E_{t-1} |$ increase or decrease the statistic is affected leading to the same types of errors.
\end{proof}

These theorems show that dependency on edge counts can lead to both false positives and false negatives from anomaly detectors that look for unusual network properties other than edge count.  In order to distinguish between the observed edge counts in each time step and the other network properties, we need a more specific data model to represent the network generation process with components for each.




\section{Density Independence} 


Now that we have established the problems with density dependence, we need to define the properties that we would prefer our network statistics to have.  To do this we need a more detailed model of how the graph examples were generated.

\subsection{Data Model}

Let the number of edges in any time step $| E_t |$ be a random variable drawn from distribution $M^{n}_E(t)$ in times where there is a normal message volume and distribution $M^{a}_E(t)$ in times where there is anomalous message volume.  Now let the distribution of edges amongst the nodes of the graph be represented by a $N \times N$ matrix $P_t$ where the value of any cell $p_{ij,t}$ is the probability of observing a message between two particular node pairs at a particular time.  This is a probability distribution so the total matrix mass sums to 1.  Like edge count, treat this matrix as drawn from distribution $M^{n}_P(t)$ in normal times and $M^{a_k}_P(t)$ in anomalous times where $k$ is the network parameter that is anomalous (for example, an atypical degree distribution).  Any observed network slice can be treated as having been generated by a multinomial sampling process where $| E_t |$ edges are selected independently from $NxN$ with probabilities $P_t$.  Denote the sampling procedure for a graph $G_t$ with $F(| E_t |, P_t)$.  In the next section, we will detail how this decomposition into the count of edges and their distribution allows us to define statistics which are not sensitive to the number of edges in the network.


\subsection{Density Consistency and Unbiasedness}
In the above data model, $P_t$ is the distribution of edges in the network, thus any property of the network aside from the volume of edges is encapsulated by the $P_t$ matrix.  Therefore, a network statistic designed to capture some network property other than edge count should be a function of $P_t$ alone.  

Let $S_k(P_t)$ be some test statistic designed to capture a network property $k$ of $P_t$, that is independent of the density of $G_t$.  Since $P_t$ is not directly observable we can estimate the statistic with the empirical statistic $\hat{S}_k(G_t)$ where $G_t$ is used to estimate $P_t$, with $\hat{p}_{ij,t} = \frac{e_{ij,t}}{| E_t |}$.  

\vspace{-3mm}
\begin{definition}
A statistic $S$ is \textbf{density consistent} if $\hat{S}(G_t)$ is a consistent estimator of $S(P_t)$.  
\end{definition}
\vspace{-3mm}

If $\hat{S}_k(G_t)$ is a consistent estimator of the true value of $S_k(P_t)$, then observing more edge examples should cause the estimated statistics to converge to the true value given $P_t$.  More specifically it is asymptotically consistent with respect to the true value as the number of observed edges increases.  Another way to describe this property is that $\hat{S}_k(G_t)$ has some bias term dependent on the edge count $| E_t |$, but the bias converges to zero as the edge count increases: $\lim_{| E_t | \rightarrow \infty} \hat{S}_k(G_t) - S_k(P_t) = 0$.  Density consistent statistics allow us to perform accurate hypothesis tests as long as a sufficient number of edges are observed in the networks.  To begin, we will prove that the rate of false positives does not exceed the selected p-value $\alpha$.

\vspace{-1mm}
\begin{theorem} \textbf{False Positive Rate for Density Consistent Statistics} \\
As $| E_t |\rightarrow \infty$ if $\hat{S}_k(G_t \!=\! F(| E_t |, P_t \sim M^{n}_P(t) ) $ converges to $S_k(P_t)$, the probability of a false positive when testing a time with an edge count anomaly $G_{t_{test}} \!=\!  F(| E_t |, P_t \sim M^{n}_P(t) )$ approaches $\alpha$.
\label{convergetrueneg}
\end{theorem}
\vspace{-3mm}


\begin{proof}
Let all learning set graphs $G_{t_{Lx}} = F(| E_{t_{Lx}} | \sim M^{n}_E({t_{Lx}}), P_{t_{Lx}} \sim M^{n}_P({t_{Lx}}))$ be drawn from non-anomalous $| E |$ and $P$ distributions and the test instance $G_{t_{test}} \break = F(| E_{t_{test}} | \sim M^{a}_E({t_{test}}), P_{t_{test}} \sim M^{n}_P({t_{test}}))$ be drawn from an anomalous $| E |$ distribution but a non-anomalous $P$ distribution.  If $\hat{S}_k(G_{t})$ is a consistent estimator of $S_k(P_{t})$, $\lim_{| E_{t} | \rightarrow \infty}  \hat{S}_k(G_{t}) = S_k(P_{t})$.  Then as both $| E_{t_{Lx}} |$ and $| E_{t_{test}} |$ increase, all learning set instances and the test set instance approach the distribution $S_k(P_{t} \sim M^{n}_P({t}))$.  As any threshold is as likely to reject a learning set instance as the test instance, the false positive rate approaches $\alpha$.
\end{proof}
\vspace{-2mm}

As the bias converges to zero, graphs created with the same underlying properties will produce statistic values within the same distribution, making the test case come from the same distribution as the null.  Even if the test case has an unusual number of edges, as long as the number of edges is not too small there will not be a false positive.  Density consistency is also beneficial in the case of false negatives.

\begin{theorem} \textbf{False Negative Rate for Density Consistent Statistics} \\
As $| E_t | \!\rightarrow\! \infty$ let $\hat{S}_k(G_t \!=\! F(| E_t |, P^{n}_t \!\sim\! M^{n}_P(t))) $ converge to $S_k(P^{n}_t)$  and $\hat{S}_k(G_t \!=\! F(| E_t |, P^{a}_t \sim M^{a}_P(t)) $ converge to $S_k(P^{a}_t)$.  If $S_k(P^{n}_t)$ and $S_k(P^{a}_t)$ are separable then the probability of a false negative $\rightarrow$ 0 as $| E_t |$ increases.
\label{convergetruepos}
\end{theorem}

\begin{proof}
Let all learning set graphs $G_{t_{Lx}} = F(| E_{t_{Lx}} | \sim M^{n}_E({t_{Lx}}), P_{t_{Lx}} \sim M^{n}_P({t_{Lx}}))$ be drawn from non-anomalous $| E |$ and $P$ distributions and the test instance $G_{t_{test}} = F(| E_{t_{test}} | \sim M^{n}_E({t_{test}}), P_{t_{test}} \sim M^{a_k}_P({t_{test}}))$ be drawn from a non-anomalous $| E |$ distribution but a $P$ distribution that is anomalous on the network property being tested.  Let $\hat{S}_k(G_{t})$ be a consistent estimator of $S_k(P_{t})$.  If $S_k(P_t \sim M^{a_k}_P({t}))$ and $S_k(P_t \sim M^{n}_P({t}))$ are separable, then \[\lim_{| E_{t_{test}} | \rightarrow \infty}  \hat{S}_k(G_{t_{test}})\] and \[\lim_{| E_{t_{Lx}} | \rightarrow \infty}  \hat{S}_k(G_{t_{Lx}})\] converge to two non-overlapping distributions and the probability of rejection approaches 1 for any $\alpha$.
\end{proof}

\vspace{-2mm}

The statistical power of a density consistent statistic depends only on whether the $P$ matrices of normal and anomalous graphs are separable using the true statistic value: as long as $| E_t |$ is sufficiently large the bias is small enough that it is not a factor in the rate of false negatives.


A special case of density consistency is {\em density consistent and unbiased}, which refers to statistics where in addition to consistency the statistic is also an unbiased statistic of the true $S_k(P_t)$.  

\begin{definition}
A statistic $S$ is \textbf{density unbiased} if $\hat{S}(G_t)$ is a unbiased estimator of $S(P_t)$.  
\end{definition}

Unbiasedness is a desirable property because a density consistent statistic without it may produce errors due to bias when the number of observed edges is low.
\subsection{Proposed Density-Consistent Statistics}


We will now define a set of Density-Consistent statistics designed to measure network properties similar to the previously described dependent statistics, but without the sensitivity to total network edge count.

\noindent \textbf{Probability Mass Shift}: \\
The probability mass shift (MS) is a parallel to GED as a measure of how much change has occurred between the two networks examined.  Mass Shift, however, attempts to measure the change in the underlying $P$ distributions and avoids being directly dependent on the edge counts.  The probability mass shift between time steps $t$ and $t-1$ is 

\vspace{-5mm}
\begin{small}
\begin{align} \label{eq:ms}
MS(P_t) = \sum_{ij} ( {p}_{ij,t} - {p}_{ij,t-1})^2
\end{align}
\end{small}
\vspace{-5mm}

\noindent The MS can be thought of as the total edge probability that changes between the edge distributions in $t$ and $t-1$.  

\noindent \textbf{Probabilistic Degree Shift}: \\
We will now propose a counterpart to the degree distribution which is density consistent.  Define the probabilistic degree of a node to be $PD(v_i) = \sum_{j \in V_{t}, i \ne j} p_{ij,t}$.  Then, let the probabilistic degree shift of a particular node in $G_t$ be defined as the squared difference of the probabilistic degree of that node in times $t$ and $t-1$.  The total probabilistic degree shift (DS) of $G_t$ is then:

\vspace{-5mm}
\begin{small}
\begin{align} \label{eq:ds}
DS(P_t) = \sum_i (\sum_{j \neq i} p_{ij,t} - \sum_{j \neq i} p_{ij,t-1})^2
\end{align}
\end{small}
\vspace{-5mm}

\noindent This is a measure of how much the total probability mass of nodes in the graph change across a single time step.  If the shape of the degree distribution is changing, the probabilistic degree of nodes will be changing as well.

\noindent \textbf{Triangle Probability}: \\
As the name suggests, the triangle probability (TP) statistics is an approach to capturing the transitivity of the network and an alternative to traditional clustering coefficient measures.  Define the triangle probability as: 

\vspace{-5mm}
\begin{small}
\begin{align} \label{eq:tp}
TP(P_t) = \sum_{i,j,k \in V, i \ne j \ne k}  p_{ij,t}  * p_{ik,t} * p_{jk,t}
\end{align}
\end{small}
\vspace{-5mm}

\section{Properties of Network Statistics} 

Now that we have described the different categories of network statistics and their relationship to the network density we will characterize several common network statistics as density dependent or consistent, comparing them to our proposed alternatives. Table~\ref{tab:statistics} summarizes our findings.

\subsection{Graph Edit Distance}

\vspace{4mm}
\begin{claim} 
GED is a density dependent statistic.
\end{claim}
\vspace{-4mm}

When the edge counts of the two time steps are the same, the GED (Eq.~\ref{eq:ged}) can be thought of as the difference in the distribution of edges in the network.  However, the GED is sensitive to $| E_t |$ in two ways: the change in the number of edges from $t\!-\!1$ to $t$: $ E_\Delta = abs( | E_{t} | -  | E_{t-1} |)$, and the minimum number of edges in each time step $E_{min} = min( | E_{t-1} |,  | E_{t} | )$.  In both cases the statistic is density dependent, and in fact it {\em diverges} as the number of edges increases. The first case is discussed in Theorem \ref{GEDdelta}, the second in Theorem \ref{GEDemin}.  

\begin{theorem} 
As $E_\Delta \rightarrow \infty$, $GED(G_t) \rightarrow \infty$, regardless of $P_t$ and $P_{t-1}$.
\label{GEDdelta}
\end{theorem}
\vspace{-5mm}


\begin{proof}
Let $E_\Delta$ be $abs(| E_t | - | E_{t-1} |)$.  Since the GED corresponds to $| E_t | + | E_{t-1} | - 2*| E_t \cap E_{t-1} |$, the minimum edit distance between $G_t$ and $G_{t-1}$ occurs when their edge sets overlap maximally and is equal to $E_\Delta$.
Therefore, as $E_\Delta$ increases even the minimum (i.e., best case) $GED(G_t,G_{t-1})$ also increases.
\end{proof}

\begin{theorem} 
As $E_{min} \!\rightarrow\! \infty$, $GED(G_t) \rightarrow \infty$ if $P_t \!\neq\! P_{t-1}$.
\label{GEDemin}
\end{theorem}
\vspace{-5mm}


\begin{proof}
Let $G_t = F(| E_{t} |, P^a_{t})$ and $G_{t-1} = F(| E_{t-1} |, P^b_{t-1})$, with $P^a_{t} \!\neq\! P^b_{t-1}$.  Select two nodes $i,j$ such that $p^a_{ij,t} \neq p^b_{ij,t-1}$.  The edit distance contributed by those two nodes is $abs( e_{ij,t} - e_{ij,t-1})$.  Let $E_{min}$ increase but $E_\Delta$ remain constant.  As $E_{min}$ increases the edit distance of the two nodes converges to $abs(E_{min}\!*\!p^a_{ij,t} - E_{min}\!*\!p^b_{ij,t-1}) = E_{min}*abs(p^a_{ij,t} - p^b_{ij,t-1})$.  Since every pair of nodes with differing edge probabilities in the two time steps will have increasing edit distance as $E_{min}$ increases, the global edit distance will also increase.
\end{proof}

\vspace{-3mm}
Since the GED measure is the literal count of edges and nodes that differ in each graph, the statistic is dependent on the difference in size between the two graphs.  Even if the graphs are the same size, comparing two large graphs is likely to produce more differences than two very small graphs due to random chance.  In addition, even when small non-anomalous differences occur between the probability distributions of edges in two time steps, variation in the edge count can result in large differences in GED.

\begin{figure}[t]
\hspace{-2mm}
{\footnotesize
\begin{tabular}{l | c | c | c |}
& Dependent & Consistent & Unbiased  \\
\hline
Graph edit distance & \checkmark && \\
Degree distribution & \checkmark && \\
Barrat clustering &  & \checkmark &  \\
\hline
Mass shift && \checkmark & \\
Degree shift && \checkmark & \\
Triangle probability &&\checkmark & \checkmark \\
\end{tabular}
}\caption{Statistical properties of previous network statistics and our proposed alternatives.}
\label{tab:statistics}
\end{figure}

\subsection{Degree Distribution Difference}

\vspace{4mm}
\begin{claim} 
DD is a density dependent statistic.
\end{claim}
\vspace{-4mm}

The degree distribution is naturally very dependent on the total degree of a network: the average degree of nodes is larger in networks with many edges.  The DD measure (Eq.~\ref{eq:dd}) is again sensitive to $| E_t |$ via $ E_\Delta$ and $E_{min}$.  In both cases the statistic is density dependent. The first case, in Theorem \ref{DegreeDist}, shows that as $E_\Delta$ increases, the DD measure will also increase even if the graphs were generated with the same $P$ probabilities. The second case, in Theorem \ref{DegreeDistMin}, shows that as $E_{min}$ increases, small variations in $P$ will increase the DD measure.


\begin{theorem}
 As $E_\Delta \uparrow$, $DD(G_t) \uparrow$, regardless of $P_t$ and $P_{t-1}$.
 \label{DegreeDist}
 \end{theorem}
\vspace{-5mm}

\begin{proof}
Pick any node $i$ in $G_t$ and $j$ in $G_{t-1}$.  If $| E_t |$ increases, and $| E_{t-1} |$ stays the same, the expected degree of $i$ increases, while $j$ stays the same; likewise the inverse is true if $| E_{t-1} |$ increases and $| E_t |$ stays the same.  Thus, as $E_\Delta$ increases the probability of any two nodes having the same degree approaches zero, so the degree distribution difference of the two networks increases with greater $E_\Delta$.
\end{proof}


\begin{theorem}
 As $E_{Min} \uparrow$, $DD(G_t) \uparrow$ \: if \: $P_t \neq P_{t-1}$.
 \label{DegreeDistMin}
 \end{theorem}
 \vspace{-5mm}

\begin{proof}
Let $PD(i,G_t) = \sum_{k \neq i} p_{ik,t}$ be the probabilistic degree of node $i$ for $G_t = F(| E_t |, P_t)$.  Pick any node $i$ in $G_t$ and $j$ in $G_{t-1}$ such that $PD(i,G_t) \neq PD(j,G_{t-1})$.  For a constant $E_\Delta$, as $E_{min}$ increases, the expected degrees of $i$ and $j$ converge to $D_{i,t}=E_{min}*\sum_{k \neq i} p_{ik,t}$ and $D_{j,t-1}=E_{min}*\sum_{k \neq j} p_{jk,t-1}$ respectively. This means that the probability of the two nodes having the same degree approaches zero.  Since every pair of nodes with differing edge probabilities in the two time steps will have unique degrees, the degree distribution difference will also increase.
\end{proof} 
 
\vspace{-3mm}
If a node has very similar edge probabilities in matrices $P_t$ and $P_{t-1}$, when few edges are sampled it is likely to have the same degree in both time steps, and thus the impact on the degree distribution difference will be low.  However, as the number of edges increases, even small non-anomalous difference in the $P$ matrices will become more apparent (i.e., the node is likely to be placed in different bins in the degree distribution difference calculation), and the impact on the measure will be larger.

\subsection{Weighted Clustering Coefficient}

\vspace{4mm}
\begin{claim} 
CB is a density consistent statistic.
\end{claim}
\vspace{-4mm}

As shown in Theorem \ref{Barrat} below, the weighted Barrat clustering coefficient (CB, Eq.~\ref{eq:wcc}) is in fact density consistent.  However, we will show later that the triangle probability statistic is also density unbiased, which gives more robust results, even on very sparse networks. 

\begin{theorem}
$CB(G_t)$ is a consistent estimator of $CB(P_t)$, with a bias that converges to 0.
\label{Barrat}
\end{theorem}
\vspace{-5mm}


\begin{proof}
For $G_t = F(| E_t |, P_t)$ the number of edges observed any pair of nodes can be represented using a multinomial distribution $\frac{| E_t |!}{e_{ij,t}! ... e_{yz,t}! } p_{ij,t}^{e_{ij,t}} ... p_{yz,t}^{e_{yz,t}}$.  As $| E_t | \rightarrow \infty$, the rate of sampling a particular node pair $i,j$ is $p_{ij,t}$, so: \[\lim_{| E_t | \rightarrow \infty} e_{ij,t} = | E_t | * p_{ij,t}\]  
Let $CB(G_t, i)$ refer to the clustering coefficent of node $i$. Then, in the limit of $| E_t |$ it converges to: 

\vspace{-4mm}
\begin{small}
\begin{align*}
\lim_{| E_t | \rightarrow \infty} \!\!CB(G_t,i) &= \frac{\sum_{j,k} | E_t | p_{ij,t} + | E_t | p_{ik,t}}{2*k_{i,lim} * \sum_{j \neq i} | E_t | p_{ij,t}} \\ 
&*(\lim_{| E_t | \rightarrow \infty} \!\!\mathcal I[e_{ij,t} \!>\! 0] * \mathcal I[e_{ik,t} \!>\! 0] * \mathcal I[e_{jk,t} \!>\! 0]) \\
&=  \frac{\sum_{j,k} p_{ij,t} + p_{ik,t}}{2*k_{i,lim} * \sum_{j \neq i} p_{ij,t}} \\  
&* \mathcal I[p_{ij,t} \!>\! 0]  * \mathcal I[p_{ik,t} \!>\! 0] * \mathcal I[p_{jk,t} \!>\! 0] \\
\end{align*}
\end{small}
\vspace{-4mm}

\vspace{-4mm}
\noindent where $k_{i,lim} = \sum_{j \neq i} \mathcal I[p_{ij,t} > 0]$.  Since this limit can be expressed in terms of $P_t$ alone, and CB is a sum of the clustering over all nodes, the Barrat clustering coefficient is a density consistent statistic.  
\end{proof}

If we calculate the expectation of the Barrat clustering, we obtain: 
\allowdisplaybreaks{
\begin{small}
\begin{align*}
E[ CB(G_t,i) ] &= E[ \frac{1}{(k_i-1)*s_i} \sum_{j,k} \frac{e_{i,j}+e_{i,k}}{2} a_{i,j}*a_{i,k}*a_{j,k} ] \\
&= E[ \frac{1}{(k_i-1)*s_i} ] \sum_{j,k} \frac{E[e_{i,j}]+E[e_{i,k}]}{2} \\& *E[a_{i,j}]*E[a_{i,k}]*E[a_{j,k}]\\
&= E[ \frac{1}{(k_i-1)*s_i} ] \sum_{j,k} \frac{p_{ij,t}+p_{ik,t}}{2} \\& *E[a_{i,j}]*E[a_{i,k}]*E[a_{j,k}]\\
&= E[ \frac{1}{(k_i-1)*s_i} ] \sum_{j,k} \frac{p_{ij,t}+p_{ik,t}}{2}\\ &*(1-p_{ij,t}^{| E_t |})*(1-p_{ik,t}^{| E_t |})*(1-p_{jk,t}^{| E_t |})
\end{align*}
\end{small}
}
\vspace{-4mm}

As this does not simplify to the limit of CB, it is not an unbiased estimator, and is thus density consistent but not unbiased. 
Other weighted clustering coefficients are also available but they have the same properties as the Barrat statistic.

%

\subsection{Probability Mass Shift}

\vspace{4mm}
\begin{claim} 
MS is a density consistent statistic.
\end{claim}
\vspace{-4mm}

As the true $P$ is unobserved, we cannot calculate the Mass Shift (MS, Eq.~\ref{eq:ms}) statistic exactly and must use the empirical Probability Mass Shift: $\widehat{MS}(G_t) = \sum_{ij} ( \hat{p}_{ij,t} - \hat{p}_{ij,t-1})^2$.  As shown in Theorem \ref{MassShift} below, the bias of this estimator 
approaches 0 as $| E_t |$ and $| E_{t-1} |$ increase, making this a density consistent statistic.  

\begin{theorem}
$\widehat{MS}(G_t)$ is a consistent estimator of $MS(P_t)$, with a bias that converges to 0.
\label{MassShift}
\end{theorem}
\vspace{-5mm}


\begin{proof}
The expectation of the empirical Mass Shift can be calculated with

\vspace{-5mm}
\begin{small}
\begin{align*}
E[ \sum_{ij} (& \hat{p}_{ij,t} - \hat{p}_{ij,t-1})^2 ] \\ &= E[ \sum_{ij} ( \frac{e_{ij,t}}{| E_t |} - \frac{e_{ij,t-1}}{| E_{t-1} |})^2 ] \\
&= \sum_{ij} E[ \frac{e^2_{ij,t}}{| E_t |^2} + \frac{e^2_{ij,t-1}}{| E_{t-1} |^2} - 2*\frac{e_{ij,t}}{| E_t |}\frac{e_{ij,t-1}}{| E_{t-1} |} ] \\
\end{align*}
\end{small}
\vspace{-7mm}

As the expectation of $e_{ij,t}^2$ for any node pair $i,j$ can be written as:

\vspace{-7mm}
\begin{small}
\begin{align*}
E[ e^2_{ij,t} ] &= Var(e_{ij,t}) + E[ e_{ij,t} ]^2 \\&= Var( Bin(| E_t |, {p}_{ij,t}) ) + E[ Bin(| E_t |, {p}_{ij,t}) ]^2  \\
&= | E_t |*{p}_{ij,t}*(1-{p}_{ij,t}) + | E_t |^2*{p}^2_{ij,t} \\
\end{align*}
\end{small}
\vspace{-7mm}

The expected empirical mass shift can then be written as:

\vspace{-7mm}
\allowdisplaybreaks{
\begin{small}
\begin{align*}
\sum_{ij} E[ \frac{e^2_{ij,t}}{| E_t |^2} &+ \frac{e^2_{ij,t-1}}{| E_{t-1} |^2} - 2*\frac{e_{ij,t}}{| E_t |}\frac{e_{ij,t-1}}{| E_{t-1} |} ] \\
&= \sum_{ij} \frac{| E_t |*{p}_{ij,t}*(1-{p}_{ij,t})}{| E_t |^2} + \frac{| E_t |^2*{p}^2_{ij,t}}{| E_t |^2}  \\&+ \frac{| E_{t-1} |*{p}_{ij,t-1}*(1-{p}_{ij,t-1})}{| E_{t-1} |^2} \\& + \frac{| E_{t-1} |^2*{p}^2_{ij,t-1}}{| E_{t-1} |^2} - 2*\frac{| E_t |*{p}_{ij,t}}{| E_t |}\frac{| E_{t-1} |*{p}_{ij,t-1}}{| E_{t-1} |} \\
&=  \sum_{ij} {p}^2_{ij,t} - 2*{p}_{ij,t}*{p}_{ij,t-1} + {p}^2_{ij,t-1} \\& + \frac{{p}_{ij,t}(1-{p}_{ij,t})}{| E_t |} + \frac{{p}_{ij,t-1}(1-{p}_{ij,t-1})}{| E_{t-1} |} \\
&=  \sum_{ij} ( {p}_{ij,t} - {p}_{ij,t-1})^2 + \frac{{p}_{ij,t}(1-{p}_{ij,t})}{| E_t |} \\& + \frac{{p}_{ij,t-1}(1-{p}_{ij,t-1})}{| E_{t-1} |} \\
\end{align*}
\end{small}
}
\vspace{-8mm}

As the two additional bias terms converge to 0 as $| E_t |$ and $| E_{t-1} |$ increase, the empirical mass shift is a consistent estimator of the true mass shift, and is density consistent.
\end{proof}

We can improve the rate of convergence as well by using our empirical estimates of the probabilities to subtract an estimate of the bias from the statistic. We use the following {\em bias-corrected} version of the empirical statistic in all experiments:

\vspace{-5mm}
\begin{small}
\begin{align} \label{eq:ms-c}
\widehat{MS}^*(G_t) &= \sum_{ij} ( \hat{p}_{ij,t} - \hat{p}_{ij,t-1})^2 \nonumber \\ &- \frac{\hat{p}_{ij,t}(1-\hat{p}_{ij,t})}{| E_t |} - \frac{\hat{p}_{ij,t-1}(1-\hat{p}_{ij,t-1})}{| E_{t-1} |}
\end{align}
\end{small}

\subsection{Probabilistic Degree Shift}

\vspace{4mm}
\begin{claim} 
DS is a density consistent statistic.
\end{claim}
\vspace{-4mm}

Again, since the true $P$ is unobserved, we cannot calculate the Degree Shift (DS, Eq.~\ref{eq:ds}) statistic exactly and must use the empirical Probability Degree Shift: \\
$\widehat{DS}(\hat{G}_t) = \sum_i (\sum_{j \neq i} \hat p_{ij,t} - \sum_{j \neq i} \hat p_{ij,t-1})^2 $.
As shown in Theorem \ref{DegreeShift} below, the bias of this estimator 
approaches 0, making this a density consistent statistic.  

\begin{theorem}
$\widehat{DS}(G_t)$ is a consistent estimator of $DS(P_t)$, with a bias that converges to 0.
\label{DegreeShift}
\end{theorem}
\vspace{-5mm}


\begin{proof}
The expectation of the empirical degree shift can be calculated with

\vspace{-5mm}
\begin{small}
\begin{align*}
E[ \widehat{DS}(G_t) ] &= E \Big[ \sum_i (\sum_{j \neq i} \hat p_{ij,t} - \sum_{j \neq i} \hat p_{ij,t-1})^2 \Big] \\
&= E \Bigg[ \sum_i (\sum_{j \neq i} \frac{e^2_{ij,t}}{| E_t |^2} + \sum_{j \neq i} \frac{e^2_{ij,t-1}}{| E_{t-1} |^2} \\& + 2*\sum_{j,k \neq i} \frac{e_{ij,t}}{| E_t |} * \frac{e_{ik,t}}{| E_t |} + 2*\sum_{j,k \neq i} \frac{e_{ij,t-1}}{| E_{t-1} |} * \frac{e_{ik,t-1}}{| E_{t-1} |} \\& - 2*\sum_{j,k \neq i} \hat p_{ij,t} \hat p_{ik,t-1} \Bigg] \\
&= \sum_i \Bigg(\sum_{j \neq i} p_{ij,t}^2 + \sum_{j \neq i} \frac{p_{ij,t}(1-p_{ij,t})}{| E_t |}  + \sum_{j \neq i} p_{ij,t-1}^2 \\& + \sum_{j \neq i} \frac{p_{ij,t-1}(1-p_{ij,t-1})}{| E_t |} + 2*\sum_{j,k \neq i} p_{ij,t} * \hat p_{ik,t} \\& + 2*\sum_{j,k \neq i} p_{ij,t-1} * p_{ik,t-1} - 2*\sum_{j,k \neq i} p_{ij,t} p_{ik,t-1} \Bigg) \\
&= \sum_i \Big(\sum_{j \neq i} \hat p_{ij,t} - \sum_{j \neq i} \hat p_{ij,t-1}\Big)^2 + \sum_{j \neq i} \frac{p_{ij,t}(1-p_{ij,t})}{| E_t |} \\& + \sum_{j \neq i} \frac{p_{ij,t-1}(1-p_{ij,t-1})}{| E_t |} \\
\end{align*}
\end{small}
\vspace{-10mm}

Which is density-consistent because the two additional bias terms converge to 0 as $| E_t |$ increases. \end{proof}

Since the bias converges to 0, the statistic is density consistent.  By subtracting out the empirical estimate of this bias term we can hasten the convergence. We use the following {\em bias-corrected} empirical degree shift in our experiments:

\vspace{-5mm}
\begin{small}
\allowdisplaybreaks{
\begin{align} \label{eq:ds-c}
 \widehat{DS}^*(\hat{G}_t) &= \sum_i \left( (\sum_{j \neq i} \hat{p}_{ij,t} - \sum_{j \neq i} \hat{p}_{ij,t-1})^2 \nonumber \right. \\& \left. - \sum_{j \neq i} \frac{\hat{p}_{ij,t}(1-\hat{p}_{ij,t})}{| E_t |} + \sum_{j \neq i} \frac{\hat{p}_{ij,t-1}(1-\hat{p}_{ij,t-1})}{| E_t |} \right)
\end{align}
}
\end{small}

\subsection{Triangle Probability}

\vspace{4mm}
\begin{claim} 
TP is a density consistent and density unbiased statistic.
\end{claim}
\vspace{-4mm}

Again, since the true $P$ is unobserved, we cannot calculate the Triangle Probability (TP, Eq.~\ref{eq:tp}) statistic exactly and must use the empirical Triangle Probability \[ \widehat{TP}(G_t) = \sum_{i,j,k \in V, i \ne j \ne k} \frac{e_{ij,t}*e_{ik,t}*e_{jk,t}}{| E_t |^3} \] which is an unbiased estimator of the true statistic (shown below in Theorem \ref{TriangleProbability}).  This means that there is no minimum number of edges necessary to attain an unbiased estimate of the true triangle probability.

\begin{theorem}
$\widehat{TP}(G_t)$ is an consistent and unbiased estimator of $TP(P_t)$.
\label{TriangleProbability}
\end{theorem}
\vspace{-5mm}



\begin{proof}
The expectation of the empirical Triangle Probability can be written as

\vspace{-5mm}
\begin{align*}
&E[ \hat{PT}(G_t)  ] =\sum_{ijk} E[ \hat p_{ij,t} \hat p_{ik,t} \hat p_{jk,t}  ] \\
&= \sum_{ijk} E[  \frac{e_{ij,t}}{| E_t |} \frac{e_{ik,t}}{| E_t |}  \frac{e_{jk,t}}{| E_t |} ] \\
\end{align*}
\vspace{-10mm}


As the number of edges on any node pair $i,j$ can be represented with a multinomial, the expectation of each is $| E_t |*p_{ij,t}$.  This lets us rewrite the triangle probability as

\vspace{-5mm}
\begin{align*}
&= \sum_{ijk} \frac{p_{ij,t}*| E_t |}{| E_t |} \frac{p_{ik,t}*| E_t |}{| E_t |}  \frac{p_{jk,t}*| E_t |}{| E_t |} \\
&= \sum_{ijk} p_{i,j}  p_{i,k}  p_{j,k} \\
\end{align*}
\vspace{-10mm}

Therefore the empirical triangle probability is an unbiased estimator of the true triangle probability and is a density consistent statistic.
\end{proof}

\section{Experiments} 

Now that we have established the properties of density-consistent and -dependent statistics we will show the tangible effects of these properties using both synthetic datasets as well as data from real networks.  The purpose of the synthetic data experiments is to show the ability of hypothesis tests using various statistics to distinguish networks that have different distributions of edges but also a random number of observed edges.  The real data experiments demonstrate the types of events that generate anomalies as well as the characteristics of the anomalies that hypothesis tests using each statistic are most likely to find.


\subsection{Synthetic Data Experiments} 

\begin{figure*}[t!]
\subfigure[Enron email network]{\includegraphics[width=.7\columnwidth]{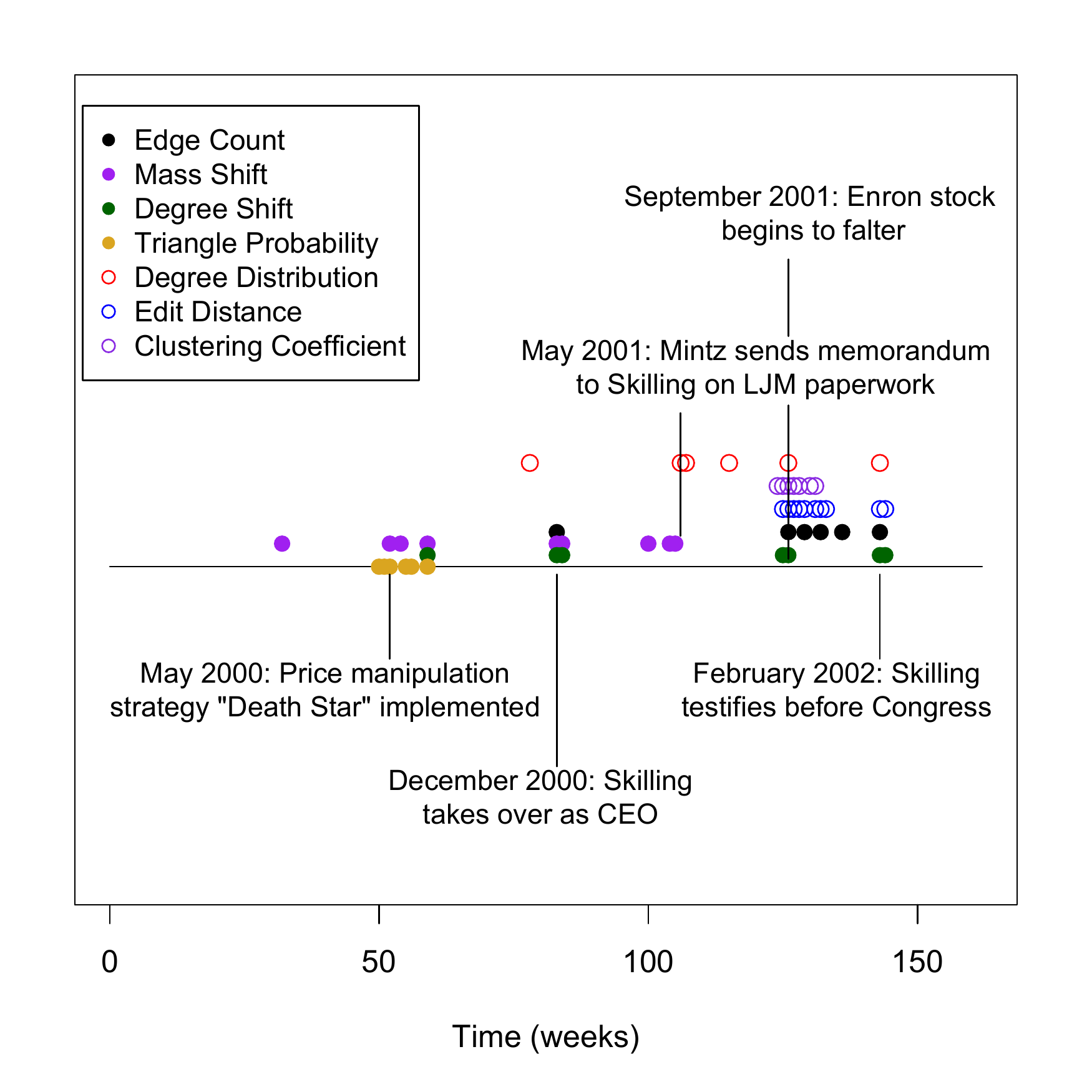} \label{enronTimeline}}
\subfigure[University e-mail network 2011-2]{\includegraphics[width=.7\columnwidth]{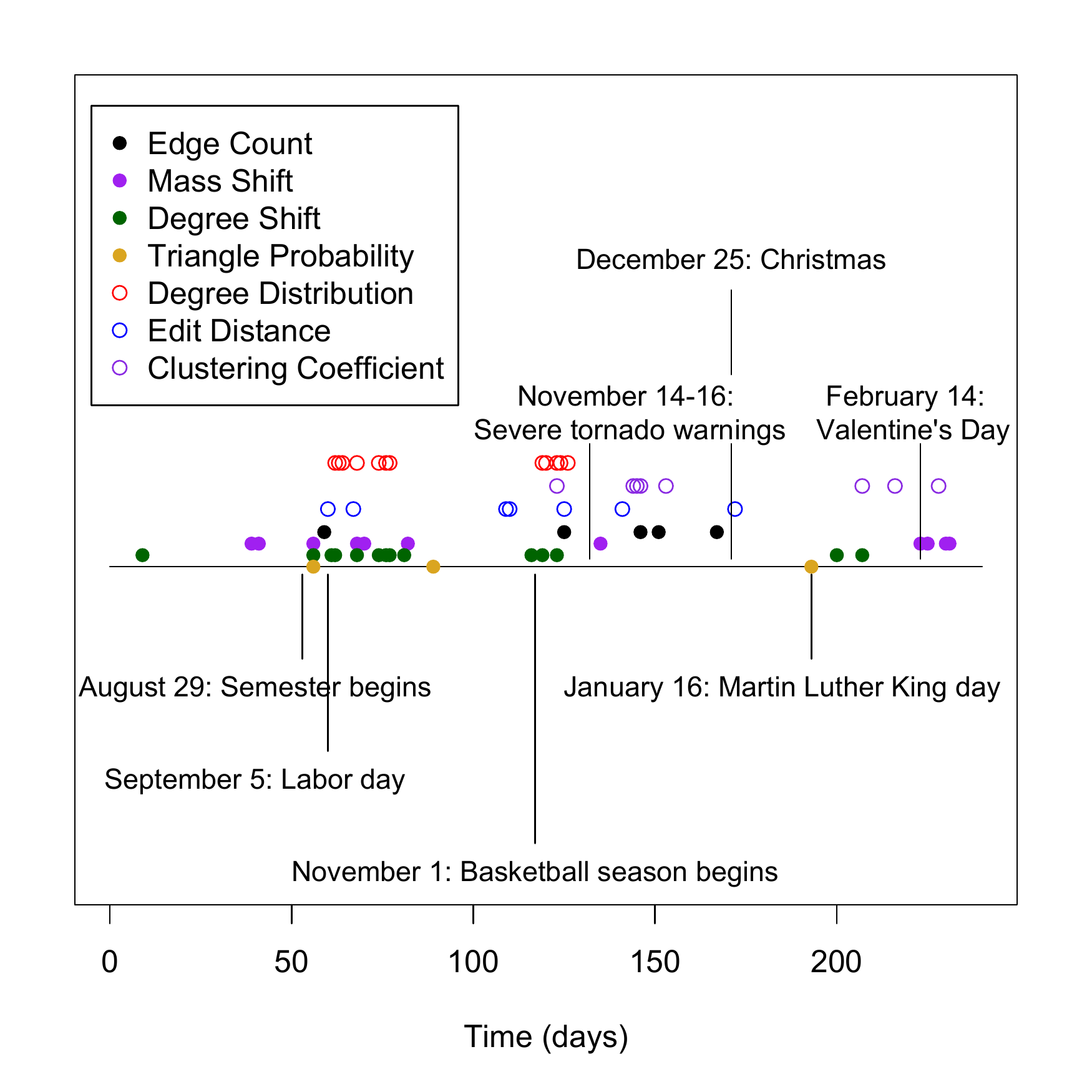} \label{purdueTimeline}}
\subfigure[Facebook Univ. subnetwork 2007-8]{\includegraphics[width=.7\columnwidth]{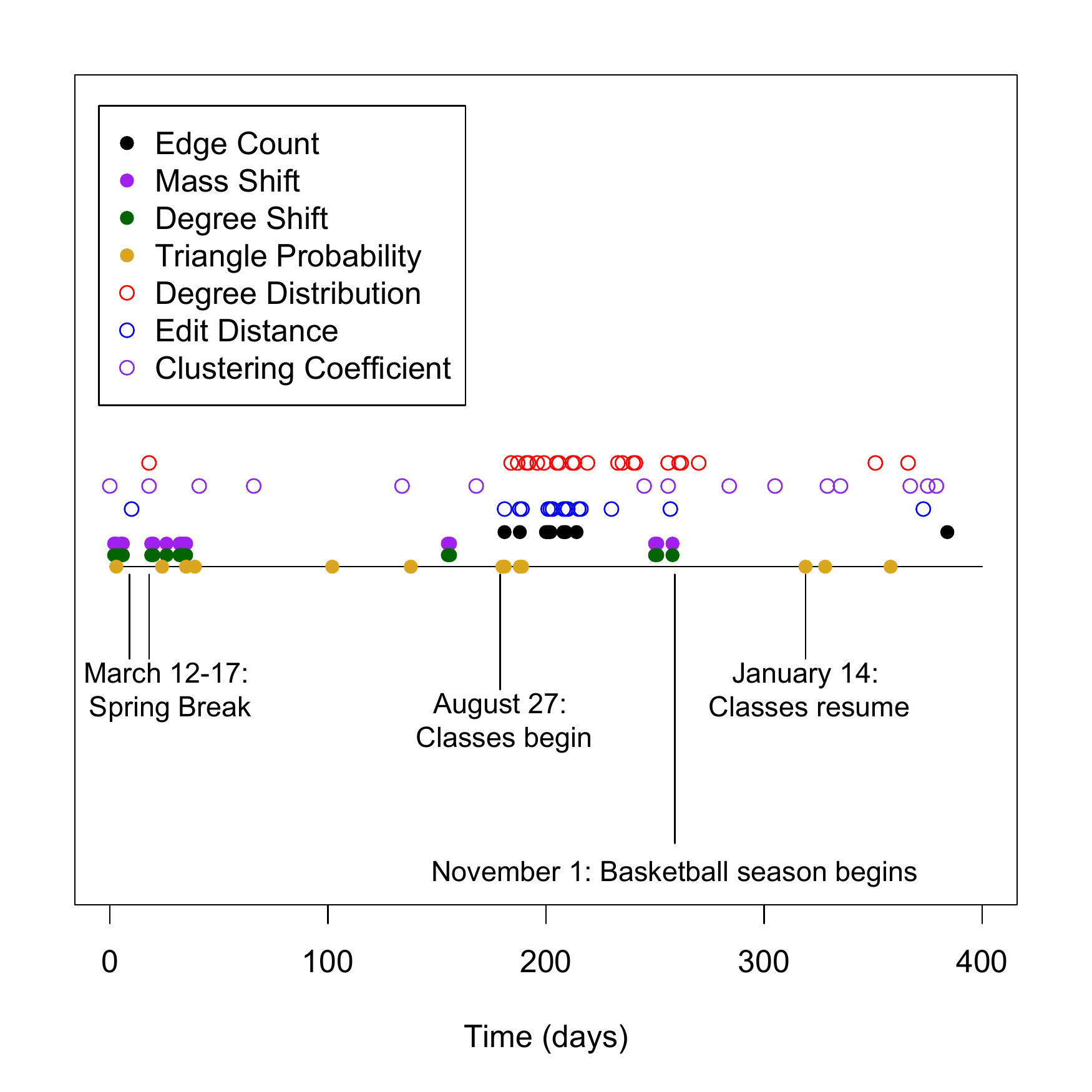}
\label{facebookTimeline}}
\vspace{-3mm}
\caption{Timelines showing reported anomalies using each statistic for three real world networks.}
\label{fig:timeline}
\end{figure*}

%

To validate the ability of Density-Consistent statistics to more accurately detect anomalies than traditional statistics we evaluated their performance on sets of synthetically generated graphs.  Rather than create a dynamic network we generated independent sets of graphs using differing model parameters and a random number of edges.  For every combination of model parameters, statistics calculated from graphs of one set of parameters (or pairs of graphs with the same parameters in the case of delta statistics like mass shift) became the null distribution, and statistics calculated on graphs from the other set of parameters (or two graphs, one from each model parameter in the case of delta statistics) became the test examples.  Treating the null distribution as normally distributed, the critical points corresponding to a p-value of 0.05 are used to accept or reject each of the test examples.  The percentage of test examples that are rejected averaged over all combinations of model parameters becomes the recall for that statistic.


To generate the synthetic graphs, we first sampled a uniform random variable representing the number of edges to insert into the graph.  For each statistic we did three sets of trials with 1k-2k, 3k-5k, and 7k-10k edges respectively.  For mass shift, edit distance, triangle probability, and clustering coefficient each synthetic graph was generated according to the edge probabilities of a stochastic blockmodel using accept/reject sampling to place the selected number of edges amongst the nodes.  Each set of models was designed to produce graphs varying a certain property, i.e. triangle probability and clustering coefficient were applied to models with varying transitivity while mass shift and edit distance skew the class probabilities by a certain amount.  



For the degree distribution and degree shift statistics, rather than using a stochastic blockmodel we assigned degrees to each node by sampling from a power-law distribution with varying parameters then using a fast Chung-Lu graph generation process to construct the network.  Unlike a standard Chung-Lu process we allow multiple edges between nodes and we continue the process until a random number of edges are inserted rather than inserting edges equal to the sum of the degrees.  The recalls are calculated the same way as with the stochastic blockmodels, using pairs of graphs with the same power-law degree distribution as the null set and differing degree distributions as the test set.  

\begin{figure}
{\footnotesize
\begin{tabular}[l]{l | c | c | c}
& 1k-2k edges & 3k-5k edges & 7k-10k edges \\
Edit Distance & 0.01 $\pm$ 0.001 & 0.05 $\pm$ 0.01 & 0.21 $\pm$ 0.02  \\
Degree Dist. & 0.15 $\pm$ 0.03 & 0.26 $\pm$ 0.05 & 0.64 $\pm$ 0.06 \\
Clust. Coef. & 0.44 $\pm$ 0.05 & 0.78 $\pm$ 0.05 & 0.78 $\pm$ 0.04  \\
\hline
Mass Shift & 0.51 $\pm$ 0.05& 0.77 $\pm$ 0.05 & 0.88 $\pm$ 0.04  \\
Degree Shift & 0.62 $\pm$ 0.07 & 0.63 $\pm$ 0.08 & 0.67 $\pm$ 0.08 \\
Triangle Prob. & 0.77 $\pm$ 0.04 & 0.94 $\pm$ 0.04 & 0.97 $\pm$ 0.03   \\
\end{tabular} 
\label{recallresults}
\caption{Recall when applying each statistic to flag synthetically generated anomalies.  Each cell is an average over all model parameter pairs.}
}
\end{figure}

\begin{figure}[t!]
\subfigure[]{\includegraphics[width=.49\columnwidth]{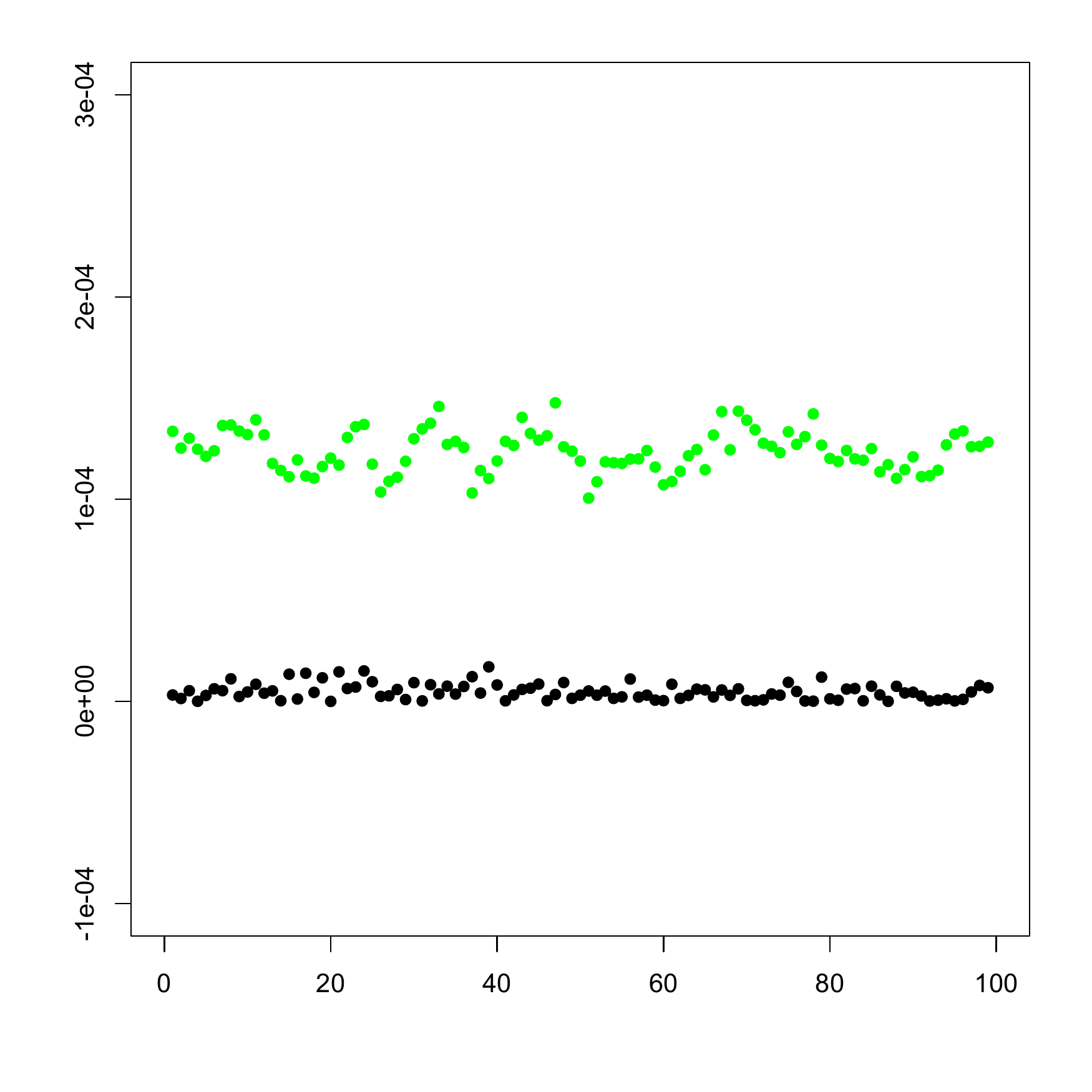}}
\subfigure[]{\includegraphics[width=.49\columnwidth]{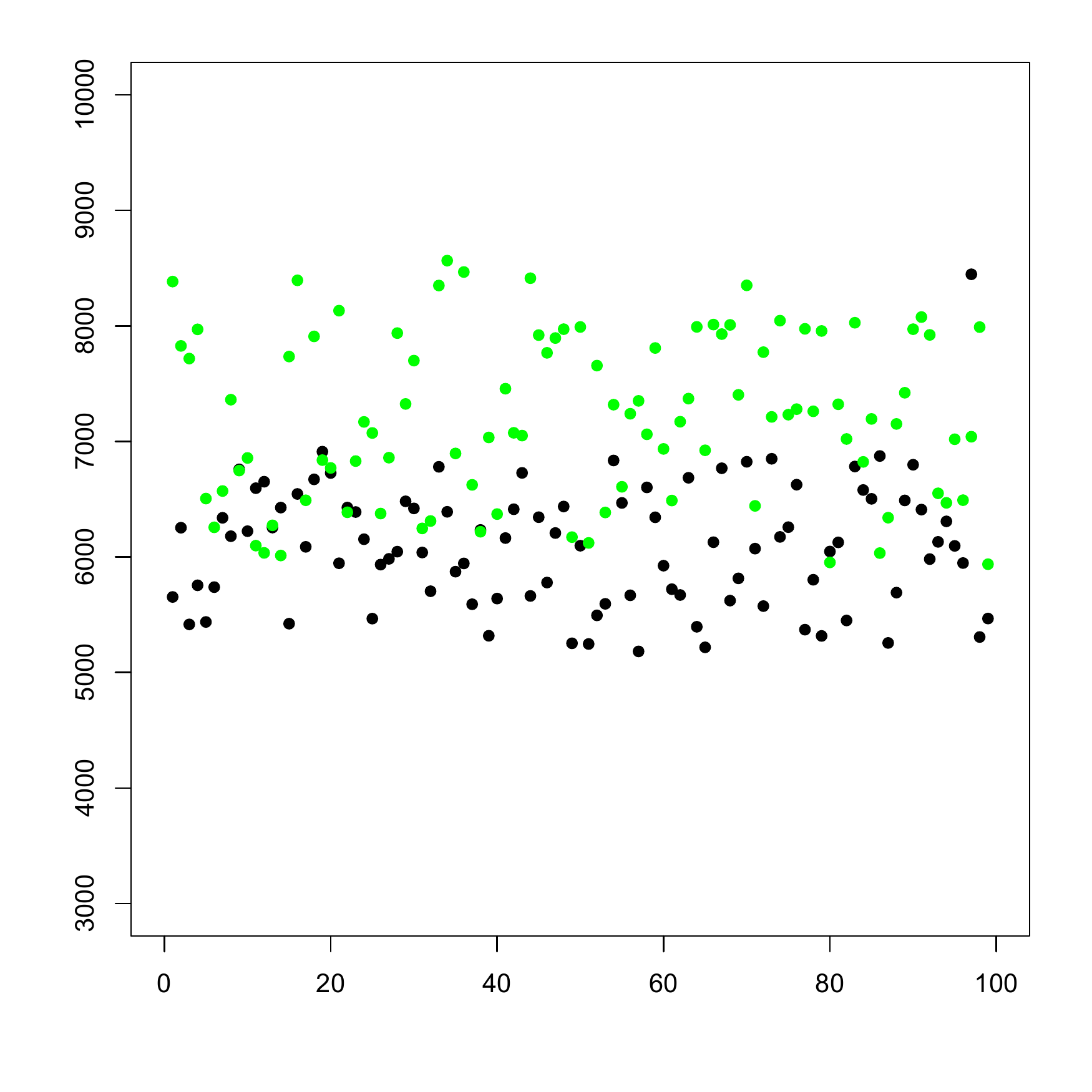}}
\caption{Comparison of Mass Shifts (a) and Edit Distances (b) produced by synthetic dataset pairs.  The variance of the Edit Distances is due to the variable edge counts of the graphs, and leads to errors when distinguishing the two distributions.} 
\label{fig:synthcomparison}
\end{figure}

Figure \ref{recallresults} shows the average recall of each statistic when applied to all models of a certain range of edges.  In general the more edges that are observed the more reliable the statistics are; however, the statistics we have proposed enjoy an advantage over the traditional statistics for all ranges of network sizes, and in some cases this improvement is as large as 200\% or more.  

Figure \ref{fig:synthcomparison} demonstrates the effect of density consistency using one pair of graph models and the mass shift and edit distance statistics.  Each black point represents a pair of datasets drawn from the same distribution while each green point is the statistic value calculated from a pair of points from differing distributions; the number of edges sampled from each distribution was 7k-10k uniformly at random.  Clearly the distribution of statistic values when calculated on graph pairs from the same distribution is different from the distribution of cross-model pairs for both statistics.  However, the randomness of the size of the graphs translates to variance in the edit distance statistics calculated leading to two distributions which are not easily separable leading to reduced recall.  Mass shift, on the other hand, is nearly unaffected by the edge variance leading to two very distinct statistic distributions.

\subsection{Real Data Experiments}

\begin{figure*}[t!]

\includegraphics[trim=0 220 0 0,clip=true,width=2.0\columnwidth]{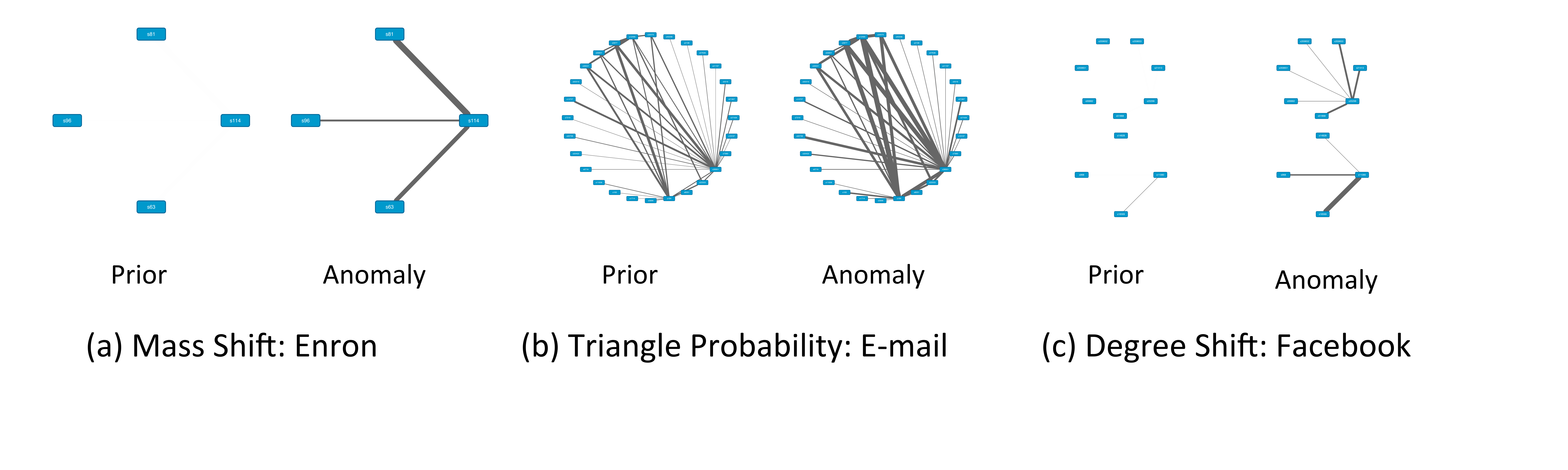}
\caption{Anomalies detected by (a) Mass Shift: Enron time steps 31-32, (b) Triangle Probability: Email time steps 88-89, and (c) Degree Shift: Facebook time steps 249-250. Each plot shows most unusual subgraph for the prior time step and the anomalous one.}
\label{fig:localanoms}
\end{figure*}

\begin{figure*}[t!]
\includegraphics[trim=0 220 0 0,clip=true,width=2.0\columnwidth]{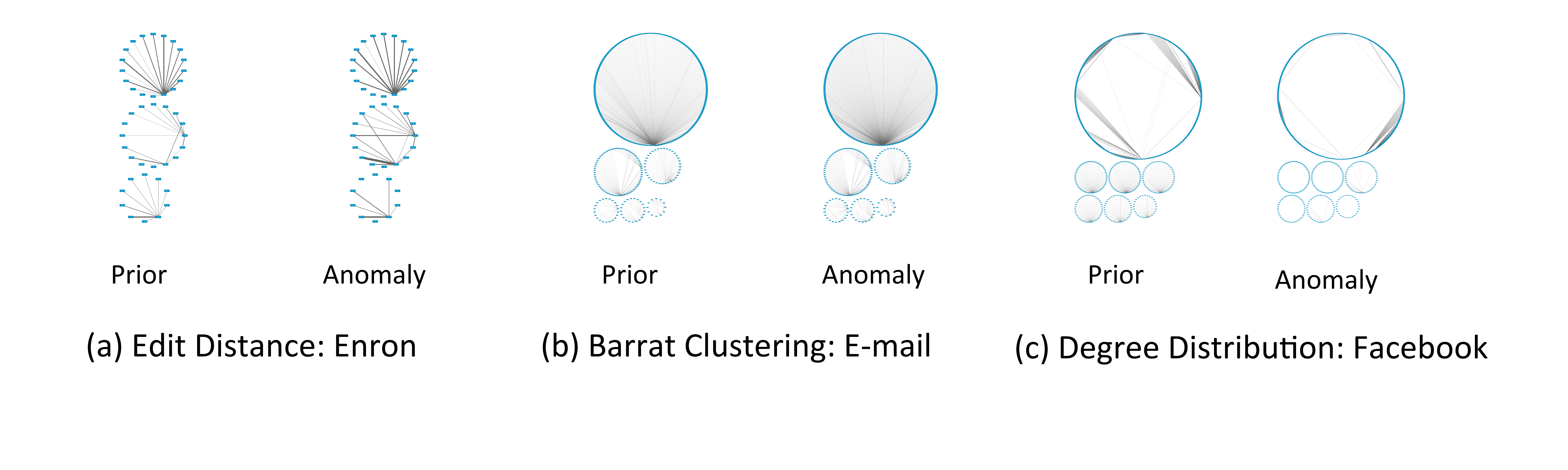}
\caption{Anomalies detected by (a) Edit Distance: Enron time steps 131-132, (b) Barrat Clustering: Email time steps 145-146, and (c) Degree Shift: Facebook time steps 365-366. Each plot shows most unusual subgraph for the prior time step and the anomalous one.}
\label{fig:localanomsother}
\end{figure*}

We will now demonstrate how using Density-Consistent statistics as the test statistics for anomaly detectors improves the ability of detectors to find novel anomalies in real-world networks.  We analyzed three dynamic network, one composed of e-mail communication from the Enron corporation during its operation and collapse, one from e-mail communications from university students, and one from the Facebook wall postings of the subnetwork composed of students at a university.  Solid points represent statistics that we have proposed in this paper while open bubbles represent classic network statistics.  Ideally the major events marked with vertical lines will be found as unusual by the detectors.  


A hypothesis test using the statistic in question is applied to every time step to determine its category as anomalous or normal.  The null distribution used is the set of statistics calculated on all other time steps of the stream.  A normal distribution is fit to the statistic values and critical points selected set according to a p-value of 0.05; any time step with a statistic that exceeds the critical points is flagged as anomalous.

Figure \ref{enronTimeline} shows the the timeline of events that occurred during the Enron scandal and breakdown.  The Density-Consistent statistics are able to recover the critical events of the timeline including events that standard statistics are not able to find.  In particular the price manipulation strategy that was implemented in the summer of 2000 and the CEO transition in December 2000 are not found by edge dependent statistics; unlike some of the later events this strategy was not accompanied with an abnormal number of edges so edge dependent statistics produced less of a signal on these points.  In addition, the time steps that are flagged by traditional statistics cluster around a set of edge count anomalies just after the Enron stock begins to crash.  This period has an elevated number of messages in general, so the statistics are responding to the number of edges in these time steps rather than changes in other properties of the network.

Figure \ref{purdueTimeline} is taken from the e-mail communication of students from the summer of 2011 to February 2012.  In general the density-consistent statistics flag time steps where certain events are taking place such as the start of basketball season, Martin Luther King Day, and Valentine's day.  The density-dependent statistics tend to flag more randomly and often coincide with times of unusual total edge count.  

Figure \ref{facebookTimeline} is constructed from the set of wall postings between members of a university Facebook group.  As with the Enron data, there is a large set of time steps with an usual total edge count that are also flagged by the density-dependent statistics.  The density-consistent statistics recover the major events such as the beginning of each semester and the start and end of spring break.  Interestingly, they also find unusual activity before and after spring break; project deadlines often fall before or after the break which may explain this activity.


Let us take a look at the network structures found in the anomalous time steps.  As the mass shift, degree shift, and triangle probability statistics are sums over the edges or nodes of the graph, by decomposing the total statistic into the values generated by each edge/node we can select a subgraph which contributes the most to the statistic.  As these statistics are designed to measure changes in the probability distribution of edges, this subgraph can be considered as the portion of the graph experiencing the most change in its edge probabilities.  

Figure \ref{fig:localanoms} shows the subgraphs of time steps that flag as anomalous in the mass shift, degree shift, or triangle probability; the subgraphs shown account for at least 50\% of the total anomaly score.  As you can see the network structures are quite different in the anomalous time steps versus the prior time steps; often nodes that do not communicate at all in the prior time step will have a significant message volume during the anomaly.  

Figure \ref{fig:localanomsother} shows the subgraphs obtained by decomposing the mass shift, degree shift, or triangle probability on time steps that were flagged as anomalies by the density-dependent statistics but not the density-consistent statistics.  Not all of the density-dependent statistics can be easily decomposed into the node/edge contributions which is why we did the density-consistent decomposition; the subgraphs generated should still be the subset of the graph experiencing the most change to its edge probabilities in that time step.  

The subgraphs discovered in these time steps tend to have less dramatic changes between the anomalous and normal time step which implies that these time steps were likely flagged due to a change in the total number of edges rather than a major shift in the edge probabilities.  Because density-consistent statistics are not sensitive to a global change in edge count they are better at detecting components of the network that have changed radically in their communication behavior.

\section{Conclusions}

In this paper we demonstrate that in order to draw proper conclusions from analysis of dynamic networks it is necessary to use methods which take into account the changes they exhibit.  As most network statistics were designed in an ad-hoc manner to describe informal properties of the network these statistics are not valid for use in a hypothesis testing approach to outlier detection when the network changes in size.  To remedy this we have described the Density-Consistency property for network statistics and shown that statistics that adhere to this property can be used for accurate anomaly detection when the network edge volume is variable.  A Density-Consistent statistic should measure some property of the network that is independent of the edge volume like the distribution of edges or the transitivity.  

We have also proposed three network statistics, Mass Shift, Degree Shift, and Triangle Probability to replace the edge dependent statistics of Graph Edit Distance, Degree Distribution and Clustering Coefficient.  We have proven that our statistics are Density-Consistent and demonstrated using synthetic trials that anomaly detectors using the consistent statistics have superior performance.  When applied to real datasets, anomaly detectors utilizing Density-Consistent statistics are often able to recover events that are missed by edge dependent statistics.

\vspace{3mm}
\bibliographystyle{abbrv}
\bibliography{WWW2015}

\end{document}